\documentclass[11pt,a4paper]{article}
\usepackage{authblk}

\usepackage{graphics}
\usepackage{graphicx}
\usepackage{amssymb,amsmath,amstext,amsfonts}

\usepackage{a4wide}

\usepackage{fullpage}
\usepackage{amsfonts,amsmath,amssymb,amsthm,boxedminipage,color,url,fullpage}
\usepackage{enumerate,paralist}
\usepackage[numbers]{natbib} 
\usepackage[labelfont=bf]{caption}
\usepackage{aliascnt,cleveref}

\usepackage{epic,eepic}

\newtheorem{lemma}{Lemma}[section]

\newtheorem{theorem}[lemma]{Theorem}

\theoremstyle{definition}

\theoremstyle{remark}

\providecommand{\keywords}[1]{\textbf{\text{Keywords:}} #1}

\begin{document}


\title{An $H_{n/2}$ Upper Bound on the Price of Stability of Undirected Network
Design Games   \thanks{This paper will appear in the Proceedings of the 39th International Symposium on Mathematical Foundations of Computer Science, MFCS 2014, Budapest, August 25-29}
}
%
\author[1]{Akaki Mamageishvili}
\author[1]{Mat\'u\v{s} Mihal\'ak}
\author[2]{Simone Montemezzani}
\affil[1]{Department of Computer Science, ETH Zurich, Switzerland}
\affil[2]{Department of Mathematics, ETH Zurich, Switzerland}
\date{June, 2014}
%

\maketitle

\begin{abstract}
In the network design game with $n$ players, every player chooses a path in an
edge-weighted graph to connect her pair of terminals, sharing costs of the edges
on her path with all other players fairly.
We study the price of stability of the game, i.e., the ratio of the social costs
of a best Nash equilibrium (with respect to the social cost) and of an optimal
play.
It has been shown that the price of stability of any network design game is at
most $H_n$, the $n$-th harmonic number. 
This bound is tight for directed graphs.
For undirected graphs, the situation is dramatically different, and tight bounds
are not known. 
It has only recently been shown that the price of stability is at
most $H_n \left(1-\frac{1}{\Theta(n^4)} \right)$, while the worst-case known
example has price of stability around 2.25. 
In this paper we improve the upper bound considerably by showing that the price
of stability is at most $H_{n/2} + \epsilon$ for any $\epsilon$ starting from
some suitable $n \geq n(\epsilon)$.
%
%
\end{abstract}


%
\section{Introduction}

Network design game was introduced by Anshelevich et al.\,\cite{original}
together with the notion of price of stability (PoS), as a formal model to study
and quantify the strategic behavior of non-cooperative agents in designing
communication networks.
Network design game with $n$ players is given by an edge-weighted graph $G$
(where $n$ does not stand for the number of vertices), and by a collection
of $n$ terminal (source-target) pairs $\{s_i,t_i\}$, $i=1,\ldots,n$.
In this game, every player $i$ connects its terminals $s_i$ and $t_i$ by an
$s_i$-$t_i$ path $P_i$, and pays for each edge $e$ on the path a fair share of
its cost (i.e., all players using the edge pay the same amount totalling to the
cost of the edge).
A Nash equilibrium of the game is an outcome $(P_1,\dots,P_n)$ in which no
player $i$ can pay less by changing $P_i$ to a different path $P_i'$.

Nash equilibria of the network design game can be quite different from an
optimal outcome that could be created by a central authority.
To quantify the difference in quality of equilibria and optima, one compares the
total cost of a Nash equilibrium to the cost of an optimum (with respect to the
total cost).
Taking the worst-case approach, one arrives at the \emph{price of anarchy},
which is the ratio of the maximum cost of any Nash equilibrium to the cost of an
optimum. 
Price of anarchy of network design games can be as high as $n$ (but not higher)
\cite{original}.
Taking the slightly less pessimistic approach leads to the notion of the
\emph{price of stability}, which is the ratio of the smallest cost of any Nash
equilibrium to the cost of an optimum.
The motivation behind this is that often a central authority exists, but cannot
force the players into actions they do not like. Instead, a central authority can
suggest to the players actions that correspond to a best Nash equilibria.
Then, no player wants to deviate from the action suggested to her, and the
overall cost of the outcome can be lowered (when compared to the worst case Nash
equilibria).

Network design games belong to the broader class of congestion games for which a
function (called a \emph{potential function}) $\Phi (P_1,\ldots,P_n)$
exists, with the property that $\Phi(P_1\ldots,P_i,\ldots,P_n) -
\Phi(P_1,\ldots,P_i',\ldots,P_n)$ exactly reflects the changes of the cost of
any player $i$ switching from $P_i$ to $P_i'$.
This property implies that a collection of paths $(P_1,\ldots,P_n)$ minimizing
$\Phi$ necessarily needs to be a Nash equilibrium.
Up to an additive constant, every congestion game has a unique potential
function of a concrete form, which can be 
used to show that the price of stability of any network design game is at most
$H_n:=\sum_{i=1}^n \frac{1}{i}$, the $n$-th harmonic number, and this is tight
for directed graphs (i.e., there is a network design game for which the price
of stability is arbitrarily close to $H_n$) \cite{original}.

Obtaining tight bounds on the price of stability for undirected graphs turned
out to be much more difficult. The worst case known example is an involved
construction of a game by Bil\`o et al.\,\cite{lowerbounds} achieving in the
limit the price of stability of around 2.25.
While the general upper bound of $H_n$ applies also for undirected graphs, it
has not been known for a long time whether it can be any lower, until the
recent work of Disser et al.\,\cite{Matus} who showed that the price of
stability of any network design game with $n$ players is at most $H_{n} \cdot
\left(1 - \frac{1}{\Theta(n^4)} \right)$. 
Improved upper bounds have been obtained for special cases.
For the case where all terminals $t_i$ are the same, Li showed \cite{multicast}
that the price of stability is at most $O\left(\frac{\log{n}}{\log\log n}
\right)$ (note that $H_n$ is approximately $\ln n$).
If, additionally, every vertex of the graph is a source of a player, a series of
papers by Fiat et al.\,\cite{Fiat+etal/2006}, Lee and Ligett
\cite{Lee+Ligett/2013}, and Bil\`o et al.\,\cite{constant} showed that the price
of stability is in this case at most $O(\log \log n)$, $O(\log \log \log n)$,
and $O(1)$, respectively.
Fanelli et al.\,\cite{Fanelli+Leniowski+Monaco+Sankowski/2012} restrict the
graphs to be rings, and prove that the price of stability is at most $3/2$.
Further special cases concern the number of players. Interestingly, tight bounds
on price of stability are known only for $n=2$ (we do not consider the case $n=1$
as a game) \cite{original,23players}, while for already 3 players there are no
tight bounds; for the most recent results for the case $n=3$, see \cite{Matus}
and \cite{bilo2011bounds}. 

All obtained upper bounds on the price of stability use the potential function
in one way or another.
Our paper is not an exception in that aspect. Bounding the price of stability
translates effectively into bounding the cost of a best Nash equilibrium. A
common approach is to bound this cost by the cost of the potential function
minimizer $(P_1^\Phi,\ldots,P_n^\Phi) := \arg \min_{(P_1,\ldots,P_n)}
\Phi(P_1,\ldots,P_n)$, which is (as we argued above) also a Nash equilibrium.
Using just the inequality $\Phi(P_1^\Phi,\ldots,P_n^\Phi) \leq
\Phi(P_1^O,\ldots,P_n^O)$, where $(P_1^O,\ldots,P_n^O)$ is an optimal outcome
(minimizing the total cost of having all pairs of terminals connected), one
obtains the original upper bound $H_n$ on the price of stability \cite{original}. In \cite{Matus,23players} authors 
consider other inequalities obtained from the property that potential optimizer is also a Nash equilibrium to obtain improved upper bounds. 
In this paper, we consider $n$ different specifically chosen
strategy profiles $(P_1^i,\ldots,P_n^i)$, $i=1,\ldots,n$, in which players use
only edges of the optimum $(P_1^O,\ldots,P_n^O)$ and of the Nash equilibrium
$(P_1^\Phi,\ldots,P_n^\Phi)$. This idea is a generalization of the
approach used by Bil\`{o} and Bove \cite{bilo2011bounds} to prove an upper bound
of $286/175 \approx 1.634$ for Shapley network design games with $3$ players.
Clearly, the potential of each of the considered strategy profile is at least
the potential of $(P_1^\Phi,\ldots,P_n^\Phi)$. Summing all these $n$
inequalities and combining it with the original inequality
$\Phi(P_1^\Phi,\ldots,P_n^\Phi) \leq \Phi(P_1^O,\ldots,P_n^O)$ gives an
asymptotic upper bound of $H_{n/2}+\epsilon$ on the price of stability.
Our result thus shows that the price of stability is strictly lower than $H_n$
by an additive constant (namely, by $\log 2$).

Albeit the idea is simple, the analysis is not. It involves carefully chosen
strategy profiles for various possible topologies of the optimum solution. These
considerations can be of independent interest in further attempts to improve the
bounds on the price of stability of network design games.


\section{Preliminaries}

\emph{Shapley network design game} is a strategic game of $n$ players played on
an edge-weighted graph $G=(V,E)$ with non-negative edge costs $c_e$, $e\in E$. Each
player $i$, $i=1,\ldots,n$, has a \emph{source} node $s_i$ and a \emph{target}
node $t_i$. All $s_i$-$t_i$ paths form the set $\mathcal{P}_i$ of the
\emph{strategies} of player $i$. 
A vector $P=(P_1,\ldots,P_n)\in \mathcal{P}_1 \times \cdots \times
\mathcal{P}_n$ is called a \emph{strategy profile}.
Let $E(P):=\bigcup_{i=1}^{n} P_i$ be the set of all edges used in $P$.
The \emph{cost of player $i$} in a strategy profile $P$ is $\text{cost}_i(P) =
\sum_{e\in P_i} c_e/k_e(P)$, where $k_e(P)=\left\vert \{j| e\in P_j \}
\right\vert$ is the number of players using edge $e$ in $P$.
A strategy profile $N=(N_1,\ldots,N_n)$ is a \emph{Nash equilibrium} if no
player $i$ can unilaterally switch from her strategy $N_i$ to a different
strategy $N_i'\in \mathcal{P}_i$ and decrease her cost, i.e., 
$\text{cost}_i(N) \leq cost_i(N_1,\ldots,N_i',\ldots,N_n)$ for every $N_i'\in
\mathcal{P}_i$.

Shapley network design games are exact potential games. That is, there is a so
called \emph{potential function} $\Phi: \mathcal{P}_1 \times \cdots \times
\mathcal{P}_n \rightarrow \mathbb{R}$ such that, for every strategy profile $P$,
every player $i$, and every alternative strategy $P_i'$, $cost_i(P) -
cost_i(P_1,\ldots,P_i',\ldots,P_n) = \Phi(P) - \Phi(P_1,\ldots,P_i',\ldots,P_n)$. 
Up to an additive constant, the potential function is unique
\cite{Monderer+Shapley/1996}, and is defined as
\begin{equation*}
  \Phi(P) = \sum_{e\in E(P)} \sum_{i=1}^{k_e(P)} c_e/i = \sum_{e\in E(P)}
  H_{k_e(P)} \, c_e \enspace\text{.}
\end{equation*}
To simplify the notation (e.g., to avoid writing $H_{\lceil n/2\rceil}$), we
extend $H_k$ also for non-integer values of $k$ by setting $H(k) := \int_0^1
\frac{1 - x^k}{1 - x} dx$, which is an increasing function, and which agrees
with the (original) $k$-th harmonic number whenever $k$ is an integer.

The \emph{social cost} of a strategy profile $P$ is defined as the sum of the
player costs:
\begin{equation}
\text{cost}(P)=\sum_{i=1}^{n} \text{cost}_i (P) = \sum_{i=1}^{n}
\sum_{e\in P_i} c_e/k_e(P) = \sum_{e\in E(P)} k_e(P) \, c_e/k_e(P) = \sum_{e\in E(P)} c_e
\text{.}
\end{equation}

%
%
A strategy profile $O(G)$ that minimizes the social cost of a game $G$ is called
a \emph{social optimum}. 
Observe that a social optimum $O(G)$ so that $E(O(G))$ induces a forest always exists (if there is a cycle,
we could remove one of its edges without increasing the social cost).
%
%
Let $\mathcal{N}(G)$ be the set of Nash equilibria of a game $G$.
The \emph{price of stability of a game} $G$ is the ratio $\text{PoS}(G) =
\min_{N\in \mathcal{N}(G)} \text{cost}(N) / \text{cost}(O(G))$.
%
%

 Let $\mathcal{M}(G)$ be the set of Nash equilibria that are also global minimizers
of the potential function $\Phi$ of the game.
The 
\emph{potential-optimal price of anarchy} of a game $G$, introduced by Kawase
and Makino \cite{popos}, 
is defined as 
$\text{POPoA}(G) = \max_{N\in \mathcal{M}(G)} \text{cost}(N) / cost(O(G))$. 
Properties of potential optimizers were earlier observed and exploited by Asadpour and Saberi in \cite{Asadpour+Saberi/2009} for other games.

Since $\mathcal{M}(G) \subset \mathcal{N}(G)$, 
it follows that $\text{PoS}(G) \leq \text{POPoA}(G)$.
%
Let $\mathcal{G}(n)$ be the set of all Shapley network design games with $n$
players. The \emph{price of stability of Shapley network design games} is
defined as $\text{PoS}(n) = \sup_{G\in \mathcal{G}(n)} \text{PoS}(G)$.
The quantity $\text{POPoA}(n)$ is defined analogously, and we get that
$\text{PoS}(n) \leq \text{POPoA}(n)$.
%


\section{The $\approx H_{n/2}$ upper bound}

The main result of the paper is the new upper bound on the price of stability,
as stated in the following theorem.

\begin{theorem}
  \label{mainthm}
$
    \text{PoS}(n)
 \leq H_{n/2} +
    \epsilon \text{,}
$
  for any $\epsilon > 0$ given that $n\geq n(\epsilon)$ for some
  suitable $n(\epsilon)$.
\end{theorem}
We consider a Nash equilibrium $N$ that minimizes the potential function $\Phi$.
For each player $i$ we construct a strategy profile $S^i$ as follows.
Every player $j\neq i$, whenever possible (the terminals of players $i$ and $j$
lie in the same connected component of the optimum $O$), uses edges of $E(O(G))$
to reach $s_i$, from there it uses the Nash equilibrium strategy (a path) of
player $i$ to reach $t_i$, and from there it again uses edges of $E(O(G))$ to
reach the player $j$'s other terminal node.
From the definition of $N$, we then obtain the inequality $\Phi(N)\leq
\Phi(S^i)$. 
We then combine these $n$ inequalities in a particular way with the inequality
$\Phi(N)\leq \Phi(O(G))$, and obtain the claimed upper bound on the cost of $N$.

The proof of Theorem \ref{mainthm} is structured in the following way.
We first prove the theorem for the special case where an optimum $O(G)$ contains
an edge that is used by every player.
We then extend the proof of this special case, first to the case where $E(O(G))$
is a tree, but with no edge used by every player, and, second, to the case where
$E(O(G))$ is a general forest (i.e., not one connected component).

We will use the following notation.
For a strategy profile $P=(P_1,\ldots,P_n)$ and a set $U\subset \{1,\ldots,n\}$,
we denote by $P_U$ the set of edges $e\in E$ for which $\{j \vert e\in P_j\} =
U$ and by $P^l$ the set of edges $e\in E$ for which $\vert \{j \vert e\in P_j\}
\vert = l$.
That is, $P_U$ is the set of edges used in $P$ by exactly the players $U$, and
$P^l=\bigcup_{U \subset \{1,\ldots,n\} \atop \vert U \vert = l} P_U$ is the set
of edges used by exactly $l$ many players.
Then the edges used by player $i$ in $P$ are $\bigcup_{U\subset \{1,\ldots,n\}
\atop i\in U} P_U$.
We stress that for every player $i\in U$, the edges of $P_U$ are part of the
strategy $P_i$; this implies that, whenever $E(P)$ induces a forest, the source
$s_i$ and the target $t_i$ are in two different connected components of $E(P)
\setminus P_U$.
For any set of edges $F\subset E$, let $ \vert F  \vert := \sum_{e\in F} c_e$.
We then have, for instance, that the cost of player $i$ in $P$ is given by
$\text{cost}_i(P)=\sum_{U\subset \{1,\ldots,n\} \atop i\in U} \frac{\vert P_U
\vert}{\vert U \vert}$.

From now on, $G$ is an arbitrary Shapley network design game with $n$ players,
$N=(N_1,\ldots,N_n)$ is a Nash equilibrium minimizing the potential function and 
$O=(O_1,\ldots,O_n)$ is an arbitrary social optimum so that $E(O)$ has no cycles.

%

\subsection{Case $O^n$ is not empty}

In this section we assume that $O^n$ is not empty. In this case, $E(O)$ is
actually a tree.
Then, $E(O) \setminus O^n$ is formed by two disconnected trees, which we call
$O^-$ and $O^+$, such that each player has the source node in one tree and the
target node in the other tree (see also Fig.~\ref{pathexample}).
Without loss of generality, assume that all source nodes $s_i$ are in $O^-$.
Given two players $i$ and $j$, let $u_{i,j}$ be the first\footnote{the edges are ordered
  naturally along the path from $s_i$ to $t_i$} edge of $O_i \cap O_j$ and
$v_{i,j}$ be the last edge of $O_i \cap O_j$.
Notice that every edge between $s_i$ and $u_{i,j}$ is used in $O$ by player $i$
but not by player $j$.
That is, each edge $e$ between $s_i$ and $u_{i,j}$ satisfies $e\in O_i$ and
$e\notin O_j$, or equivalently, $e\in \bigcup_{U \subset \{1,\ldots,n\} \atop
i\in U, j\notin U} O_U$.
An analogous statement holds for each edge $e$ between $t_i$ and $v_{i,j}$.

For every player $i$, we define a strategy profile $S^i$, where
%
%
player $j=1,\ldots,n$ uses the following $s_j$-$t_j$ path $S_j^i$ (see Fig.
\ref{pathexample} for an example.):
\begin{enumerate}
  \item From $s_j$ to $u_{i,j}$, it uses edges of $O^-$.  \item From $u_{i,j}$ to $s_i$, it uses edges of $O^-$.  \item From $s_i$ to $t_i$, it uses edges of $N_i$.  \item From $t_i$ to $v_{i,j}$, it uses edges of $O^+$.  \item From $v_{i,j}$ to $t_j$, it uses edges of $O^+$.
\end{enumerate}

Observe that for $j=i$, the path $S_i^i$ is just the path $N_i$ from the Nash
equilibrium $N$. Then, $S^i = ( S^i_1,\ldots,S^i_n )$.

\begin{figure}[t]
  \centering
  \includegraphics[width=0.5\linewidth]{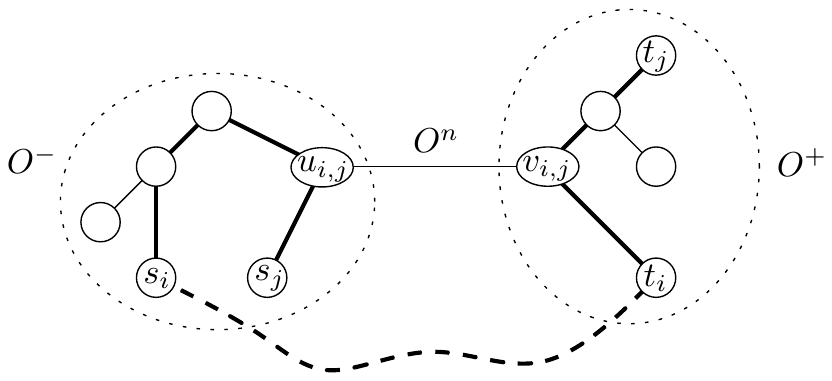}
  \caption{The non dashed lines are the edges of $E(O)$, the dashed line is the
  Nash strategy $N_i$. The path $S_j^i$ from $s_j$ to $t_j$ is given by the
  thicker dashed and non dashed lines.}
  \label{pathexample}
\end{figure}
If $S_j^i$ contains cycles, we skip them to obtain a simple path from $s_j$ to
$t_j$. This can be the case if $N_i$ is not disjoint from $E(O)$, so that an
edge appears both in step $3$ and in one of the steps $1,2,4$ or $5$.
Observe that the path $S_j^i$ uses exactly the edges of $O_U$ for $i\in U,
j\notin U$ (in steps 2 and 4), the edges of $O_U$ for $i\notin U, j\in U$ (in
steps 1 and 5) and the edges of $N_U$ for $i\in U$ (in step 3).
We now can prove the following lemma.

\begin{lemma}
  \label{firstlemma}
  For every $i\in\{1,\ldots,n\}$,
  \begin{equation}
    \label{firstequation}
    \Phi(N)\leq \Phi(S^i) \leq \sum_{U \subset \{1,\ldots,n\} \atop i\in U} H_n
    \vert N_U \vert  + \sum_{U \subset \{1,\ldots,n\} \atop i\in U} H_{n-{\vert
    U \vert}}  \vert O_U \vert + \sum_{U \subset \{1,\ldots,n\} \atop i\notin U}
    H_{\vert U \vert}  \vert O_U \vert \text{.}
  \end{equation}
\end{lemma}

\begin{proof}
The first inequality of (\ref{firstequation}) holds because, by assumption, $N$ is
a global minimum of the potential function $\Phi$.

To prove the second inequality, recall that for any strategy profile $P$ we can
write $\Phi(P)=\sum_{e\in P} H_{k_e(P)} c_e=\sum_{U \subset \{1,\ldots,n\}}
H_{\vert U \vert} \vert P_U \vert $.
In our case, every edge $e\in S^i$ belongs either to $N_U$, $U\subset
\{1,\ldots,n\}$, $i \in U$, or to $O_U$, and we therefore sum only over these
terms.
%
%
We now show that, in our sum, the cost $c_e$ of every edge $e$ in $S^i$ is
accounted for with at least coefficient $H_{k_e(S^i)}$.


For the first sum in the right hand side of (\ref{firstequation}), obviously at
most $n$ players can use an edge of $N_U, i\in U$, i.e., $k_e(S^i)\leq n$.
To explain the second and third sums, notice that if an edge $e\in O_U$ that is
present in $S^i$ also belongs to $N_i$, its cost is already accounted for in the
first sum.
So, we just have to look at edges that are only present in steps $1,2,4$ and $5$
of the definition of $S_j^i$.

To explain the second sum, let $i\in U$. Then, as we already noted, in the
definition of $S_j^i$, player $j$ uses edges of $O_U$ with $i\in U$ only if
$j\notin U$ (in steps $2$ and $4$). Since there are exactly $n-\vert U \vert$
players that satisfy $j\notin U$, this explains the second sum.

Finally, to explain the third sum, let $i\notin U$. Similarly to the previous
argument, in the definition of $S_j^i$, player $j$ uses edges of $O_U$ with
$i\notin U$ only if $j\in U$ (in steps $1$ and $5$). Since there are exactly
$\vert U \vert$ players that satisfy $j\in U$, this explains the third sum.

\end{proof}

We now show how to combine Lemma~\ref{firstlemma} with the inequality $\Phi(N)
\leq \Phi(O)$ to prove Theorem \ref{mainthm}, whenever $O^n\neq \emptyset$.


\begin{lemma}
  \label{boundlemma}
  Suppose that Inequality \eqref{firstequation} holds for every $i$. Then,
  for $x=\frac{n - H_n}{H_n - 1}$, 
  \begin{equation*}
    \text{PoS}(G)\leq \frac{n+x}{n+x-H_{n}} H_{\frac{n+x}{2}}\leq H_{n/2} +
    \epsilon
  \end{equation*}
  holds for any $\epsilon > 0$, given that $n\geq n(\epsilon)$ for some
  suitable $n(\epsilon)$.
\end{lemma}
\begin{proof}
We sum \eqref{firstequation} for $i=1,\ldots,n$ to obtain
%
\begin{equation*}
\begin{split}
& n \Phi(N) \leq \sum_{i=1}^{n} \left( \sum_{U \subset \{1,\ldots,n\} \atop i\in U} H_n \vert N_U \vert  + \sum_{U \subset \{1,\ldots,n\} \atop i\in U} H_{n-{\vert U \vert}}  \vert O_U \vert + \sum_{U \subset \{1,\ldots,n\} \atop i\notin U} H_{\vert U \vert}  \vert O_U \vert \right) = \\
& = \sum_{U \subset \{1,\ldots,n\}} \vert U \vert H_n \vert N_U \vert  + \sum_{U \subset \{1,\ldots,n\}} \vert U \vert H_{n-{\vert U \vert}}  \vert O_U \vert + \sum_{U \subset \{1,\ldots,n\}} (n - \vert U \vert) H_{\vert U \vert}  \vert O_U \vert = \\
&  = \sum_{l=1}^{n} l H_n \vert N^l \vert + \sum_{l=1}^{n} (l H_{n-l} + (n-l) H_l) \vert O^l \vert \enspace .
\end{split}
\end{equation*}
Since $\Phi(N)=\sum_{l=1}^{n} H_l \vert N^l \vert$, by putting all terms relating to $N$ on the left hand side we obtain
\begin{equation}
\label{pf1eq1}
\sum_{l=1}^{n} (n H_l- l H_n) \vert N^l \vert \leq \sum_{l=1}^{n} (l H_{n-l} + (n-l) H_l) \vert O^l \vert \enspace .
\end{equation}
On the other hand, we have $\Phi(N) \leq \Phi(O)$, which we can write as
\begin{equation}
\label{pf1eq2}
\sum_{l=1}^{n} H_l \vert N^l \vert \leq \sum_{l=1}^{n} H_l \vert O^l \vert \enspace .
\end{equation}
If we multiply (\ref{pf1eq2}) by $x=\frac{n-H_n}{H_n-1}$ and sum it with (\ref{pf1eq1}) we get
\begin{equation}
\label{pf1eq3}
\sum_{l=1}^{n} ((n+x) H_l- l H_n) \vert N^l \vert \leq \sum_{l=1}^{n} (l H_{n-l} + ((n+x)-l) H_l) \vert O^l \vert \enspace .
\end{equation}
Let $\alpha(l)=(n+x) H_l- l H_n$ and $\beta(l)=l H_{n-l} + ((n+x)-l) H_l$.
We will show that $\min_{l\in \{1,\ldots,n\}} \alpha(l)=n+x-H_{n}$ and that
$\max_{l\in \{1,\ldots,n\}} \beta(l) \leq (n+x) H_{\frac{n+x}{2}}$.
This will allow us to bound the left and right hand side of (\ref{pf1eq3}),
giving us the desired bound on the price of stability.

To prove $\min_{l\in \{1,\ldots,n\}} \alpha(l)=n+x-H_n$, we 
show that $\alpha(l)$ first increases and then decreases and that $\alpha(1)=\alpha(n)$.
We have
\begin{equation*}
\begin{split}
&\alpha(l+1)-\alpha(l) = (n+x) H_{l+1} - (l+1) H_n - ((n+x) H_l- l H_n) = \\
& = (n+x) H_l + \frac{n+x}{l+1} - l H_n - H_n - (n+x) H_l + l H_n =
\frac{n+x}{l+1} - H_n\text{.}
\end{split}
\end{equation*}
The difference is positive when $l+1\leq \frac{n+x}{H_n}$, which proves that $\alpha$ first increases and then decreases and implies that the minimum is at one of the extremes $l=1$ or $l=n$.
Is it easy to check that 
by the choice of $x$ the values at the two extremes coincide, and 
the minimum is $\alpha(1)=n+x-H_n$.

To prove $\max_{l\in \{1,\ldots,n\}} \beta(l) \leq (n+x) H_{\frac{n+x}{2}}$, we
first show that $\theta(l)=l H_{n-l} + (n-l) H_l$ has maximum $n H_{n/2}$.
Since $\theta$ is symmetric around $n/2$, we just have to show that the
difference $\theta(l+1)-\theta(l)$ is always positive for $l+1\leq n/2$.
This proves that $\theta$ reaches at $l=n/2$ the maximum value of $\frac{n}{2}
H_{n/2}+\frac{n}{2} H_{n/2}=n H_{n/2}$.
We have that
\begin{equation*}
  \begin{split}
    &\theta(l+1)-\theta(l) = (l+1) H_{n-(l+1)} + (n-(l+1)) H_{l+1} - (l H_{n-l}
    + (n-l) H_l) = \\
    & = l H_{n-l} + H_{n-l} - \frac{l+1}{n-(l+1)} + (n - l) H_l  - H_l + \frac{n
    - (l+1)}{l+1} - l H_{n-l} - (n-l) H_l = \\
    & = \frac{n - (l+1)}{l+1} - \frac{l+1}{n-(l+1)} + H_{n-l} - H_l \text{.}
  \end{split}
\end{equation*}
The term $\frac{n - (l+1)}{l+1} - \frac{l+1}{n-(l+1)}$ is positive if $n -
(l+1)\geq l+1$, that is if $l+1\leq n/2$.
Since $H$ is an increasing function, $H_{n-l} - H_l$ is positive if $l\leq n/2$,
in particular if $l+1\leq n/2$. This proves our claim that $\theta(l)=l H_{n-l}
+ (n-l) H_l$ has maximum $n H_{n/2}$.

Since $H$ is an increasing function, we then have the bound
\begin{equation*}
\beta(l)=l H_{n-l} + ((n+x)-l) H_l\leq l H_{(n+x)-l} + ((n+x)-l) H_l \leq (n+x) H_{\frac{n+x}{2}} \enspace .
\end{equation*}
We can now finally prove Lemma \ref{boundlemma}. We know that
\begin{equation}
  \label{pf1eq4}
  (n+x-H_{n}) \text{ cost}(N) = (n+x-H_{n}) \sum_{l=1}^{n} \vert N^l \vert \leq
  \sum_{l=1}^{n} ((n+x) H_l- l H_n) \vert N^l \vert \enspace \text{,}
\end{equation}
\begin{equation}
  \label{pf1eq5}
  \sum_{l=1}^{n} (l H_{n-l} + ((n+x)-l) H_l) \vert O^l \vert \leq (n+x)
  H_{\frac{n+x}{2}} \sum_{l=1}^{n} \vert O^l \vert = (n+x) H_{\frac{n+x}{2}}
  \text{ cost}(O) \enspace \text{,}
\end{equation}
which together with (\ref{pf1eq3}) proves that $\text{PoS}(G) \leq
\frac{\text{cost}(N)}{\text{cost}(O)}\leq \frac{n+x}{n+x-H_{n}}
H_{\frac{n+x}{2}}$.

Now observe that for any $\epsilon$ there is an $n(\epsilon)$ so that
$\frac{n+x}{n+x-H_{n}} H_{\frac{n+x}{2}} \leq H_{n/2} + \epsilon$ whenever $n\geq n(\epsilon)$, because
$x=\frac{n-H_n}{H_n-1} \in o(n)$ and $(H_n)^2 \in o(n)$.
%
%

\end{proof}



\subsection{Case $O^n$ is empty}

In the previous section we proved Theorem \ref{mainthm} if $O^n\neq \emptyset$
by constructing for every pair of players $i$ and $j$ a particular path $S^i_j$
that uses edges of $E(O)$ to go from $s_j$ to $s_i$ and from $t_j$ to $t_i$.

If $E(O)$ is not connected, then there is a pair of players $i,j$ for which $s_i$ and $s_j$ are in different connected components of $E(O)$, and we cannot define the path $S_j^i$.
Even if $E(O)$ is connected, but $O^n = \emptyset$, there might be a pair of players $i$ and $j$ for which the path $S^i_j$ exists, but this path is not optimal.
See Fig. \ref{case6} for an example: the path $S_j^i$ (before cycles are removed to make $S_j^i$ a simple path) traverses some edges of $E(O)$ twice, including the edge denoted by $e$ in the figure.
The same holds even if we exchange the labeling of $s_i$ and $t_i$.
Thus, we may need to define a new path $T_j^i$ for some players $i$ and $j$.



To define the new path $T_j^i$, let us introduce some notation.
Given two players $i, j$ and two nodes $x_i\in \{ s_i, t_i \}, x_j \in \{ s_j, t_j \}$ in the same connected component of $E(O)$, let $O(x_i,x_j)$ be the unique path in $E(O)$ between $x_i$ and $x_j$.
If $s_i$ and $s_j$ are in the same connected component of $E(O)$, let $(T_j^i)'$ (respectively $(T_j^i)''$) be the following $s_j$-$t_j$ path:
\begin{enumerate}
\item[$1'.$] From $s_j$ to $s_i$ (respectively $t_i$), it uses edges of $O(s_i,s_j)$ (respectively $O(t_i,s_j)$).
\item[$2'.$] From $s_i$ (respectively $t_i$) to $t_i$ (respectively $s_i$), it uses edges of $N_i$.
\item[$3'.$] From $t_i$ (respectively $s_i$) to $t_j$, it uses edges of $O(t_i,t_j)$ (respectively $O(s_i,t_j)$).
\end{enumerate}


If $(T_j^i)'$ or $(T_j^i)''$ contain cycles, we skip them to obtain a simple path from $s_j$ to $t_j$. 
See Fig. \ref{case1} for an example of $(T_j^i)'$ and Fig. \ref{case3} for an example of $(T_j^i)''$.

Notice that in the previous section, we had $S_j^i=(T_j^i)'$ (where steps $1$ and $2$ are now step $1'$; steps $4$ and $5$ are now step $3'$) and $O(s_i,s_j) \cap O(t_i,t_j) = \emptyset$, since $O(s_i,s_j) \subset O^-$ and $O(t_i,t_j) \subset O^+$.
This ensured that there was no edge that is traversed both in step $1'$ and $3'$, which would make Lemma \ref{firstlemma} not hold.
In general, $O(s_i,s_j) \cap O(t_i,t_j) = \emptyset$ does not have to hold; for example in Fig. \ref{case6} we have $e \in O(s_i,s_j) \cap O(t_i,t_j)$.
We call the path $(T_j^i)'$ (respectively $(T_j^i)''$) \emph{$O$-cycle free} if $O(s_i,s_j) \cap O(t_i,t_j) = \emptyset$ (respectively if $O(s_i,t_j) \cap O(t_i,s_j) = \emptyset$).
For instance, in Fig. \ref{case6} both $(T_j^i)'$ and $(T_j^i)''$ are not \emph{$O$-cycle free}.



We are now ready to define the path $T_j^i$ for two players $i$ and $j$.
If $s_i$ and $s_j$ are in the same connected component of $E(O)$, we set $T_j^i = (T_j^i)'$ (respectively $T_j^i= (T_j^i)''$) if $(T_j^i)'$ (respectively $(T_j^i)''$) is \emph{$O$-cycle free}.
Otherwise, we set $T_j^i=O_j$.
Similar to the previous section, let $T^i=(T_1^i,\ldots,T_n^i)$.
That is, in $T^i$ a player $j$ uses the optimal path $O_j$ if the paths $(T_j^i)'$ and $(T_j^i)''$ are not defined (meaning that $s_i$ and $s_j$ are in different connected components of $E(O)$), or if they are not \emph{$O$-cycle free} (meaning that they use some edges of $E(O)$ twice).
Otherwise, player $j$ uses the \emph{$O$-cycle free} path.

The following lemma shows that the paths $T^i$ satisfy the requirements of Lemma \ref{boundlemma} if $E(O)$ is connected but $O^n = \emptyset$.
A subsequent lemma will then show that the requirements of Lemma \ref{boundlemma} are satisfied even if $E(O)$ is not connected.

\begin{lemma}
\label{secondlemma}
If $E(O)$ is connected, then for every $i\in\{1,\ldots,n\}$
\begin{equation}
\label{secondequation}
\Phi(N)\leq \Phi(T^i) \leq \sum_{U \subset \{1,\ldots,n\} \atop i\in U} H_{n} \vert N_U \vert  + \sum_{U \subset \{1,\ldots,n\} \atop i\in U} H_{o_i(U)}  \vert O_U \vert + \sum_{U \subset \{1,\ldots,n\} \atop i\notin U} H_{\vert U \vert}  \vert O_U \vert \enspace ,
\end{equation}
with $o_i(U)\leq n-{\vert U \vert}$.
\end{lemma}
\begin{proof}
Since the initial part of the proof is exactly the same as the proof of Lemma \ref{firstlemma}, we only prove that the cost $c_e$ of every edge $e$ in $T^i$ is accounted for with at least coefficient $H_{k_e(T^i)}$ in the right hand side of (\ref{secondequation}).
In particular, we just look at edges that are only present in steps $1'$ and $3'$ of the definition of $T_j^i$, since an edge $e\in O_U$ that also belongs to $N_i$ has its cost already accounted for in the first sum.

To explain the second and third sum, let $U\subset \{1,\ldots,n\}$ and $e \in O_U$.
We will look at all the possibilities of where the nodes $s_i,s_j,t_i$ and $t_j$ can be in the tree $E(O)$ and see whether $e$ can be traversed in the path $T_j^i$.
Denote by $e^-$ and $e^+$ the two distinct connected components of $E(O) \setminus \{e\}$.
Then, by the definition of $O_U$, each player $k\in U$ has $s_k\in e^-$ and $t_k \in e^+$, or viceversa.
Always by the definition of $O_U$, each player $k\notin U$ has either $s_k,t_k \in e^-$ or $s_k,t_k \in e^+$.

To explain the third sum of (\ref{secondequation}), let $i\notin U$.
For illustration purposes, assume without loss of generality that $s_i,t_i \in e^-$.
Then, the only possibilities are that
\begin{itemize}
\item $j\in U$. Then $e$ can be traversed, since $T_j^i$ has to go from $e^-$ to $e^+$ to connect $s_j$ and $t_j$. See Fig. \ref{case1} for an illustration in the case $T_j^i \neq O_j$ and Fig. \ref{case2} for the case $T_j^i=O_j$.
\item $j\notin U, s_j,t_j \in e^-$. Then $e$ cannot be traversed, since all terminal nodes are in $e^-$ and there is no need to traverse $e$. See Fig. \ref{case3} for an illustration in the case $T_j^i \neq O_j$ and Fig. \ref{case5} for the case $T_j^i=O_j$.
\item $j\notin U, s_j,t_j \in e^+$. Then $e$ cannot be traversed, since both $(T_j^i)'$ and $(T_j^i)''$ 
 traverse $e$ twice, so we must have $T_j^i = O_j$. See Fig. \ref{case6} for an illustration.
\end{itemize}
As we can see, $e$ can be traversed only if $j\in U$, that is, at most $\vert U \vert$ times.
This explains the third sum of (\ref{secondequation}).
\begin{figure}[p!]
\centering
\begin{minipage}[h]{.48\textwidth}
\centering
\includegraphics[width=\linewidth]{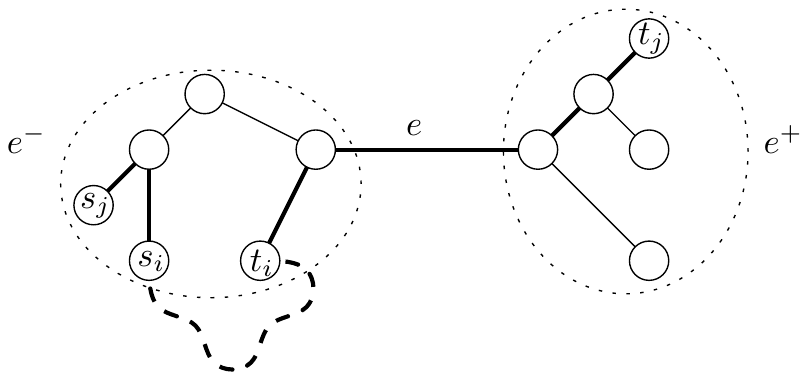}
\caption{$i\notin U, j\in U$ and $T_j^i \neq O_j$. Then $e$ can be traversed in the path $T_j^i$.}
\label{case1}
\end{minipage}%
\hfill
\begin{minipage}[h]{.48\textwidth}
\centering
\includegraphics[width=\linewidth]{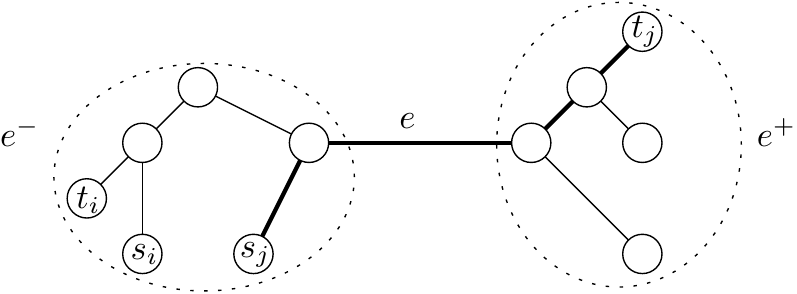}
\caption{$i\notin U, j\in U$ and $T_j^i = O_j$. Then $e$ can be traversed in the path $T_j^i$.}
\label{case2}
\end{minipage}
\end{figure}
\begin{figure}
\centering
\begin{minipage}[h]{.48\textwidth}
\centering
\includegraphics[width=\linewidth]{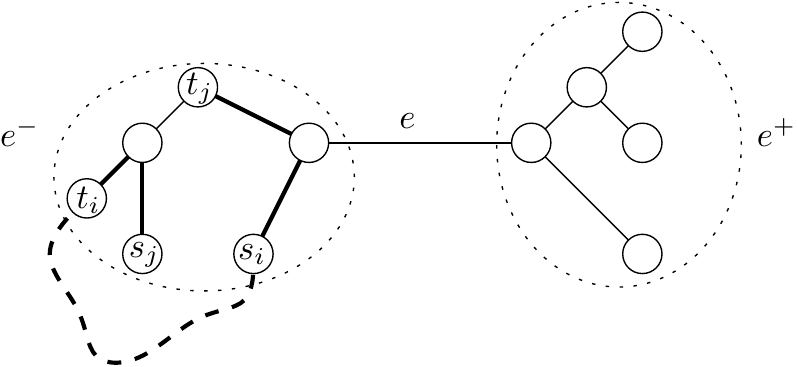}
\caption{$i\notin U, j\notin U$, $s_j,t_j \in e^-$ and $T_j^i \neq O_j$ . Then $e$ cannot be traversed in the path $T_j^i$.}
\label{case3}
\end{minipage}%
\hfill
\begin{minipage}[h]{.48\textwidth}
\centering
\includegraphics[width=\linewidth]{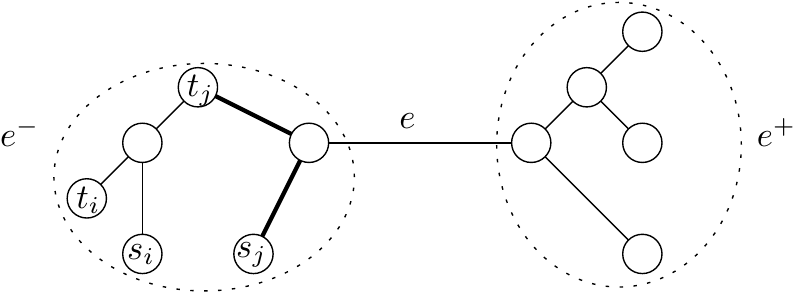}
\caption{$i\notin U, j\notin U$, $s_j,t_j \in e^-$ and $T_j^i = O_j$ . Then $e$ cannot be traversed in the path $T_j^i$.}
\label{case5}
\end{minipage}
\end{figure}
\begin{figure}
\centering
\begin{minipage}[h]{.48\textwidth}
\centering
\includegraphics[width=\linewidth]{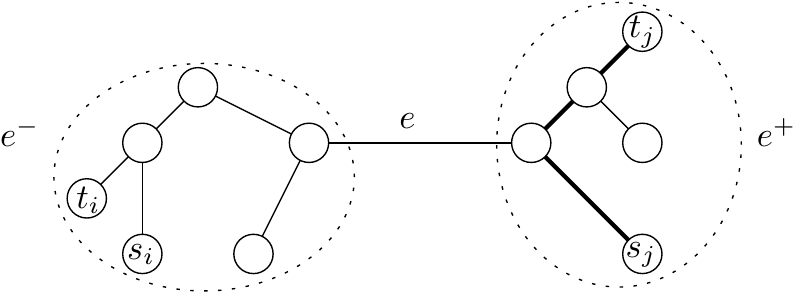}
\caption{$i\notin U, j\notin U$, $s_j,t_j \in e^+$ and $T_j^i = O_j$ . Then $e$ cannot be traversed in the path $T_j^i$.}
\label{case6}
\end{minipage}%
\hfill
\begin{minipage}[h]{.48\textwidth}
\centering
\includegraphics[width=\linewidth]{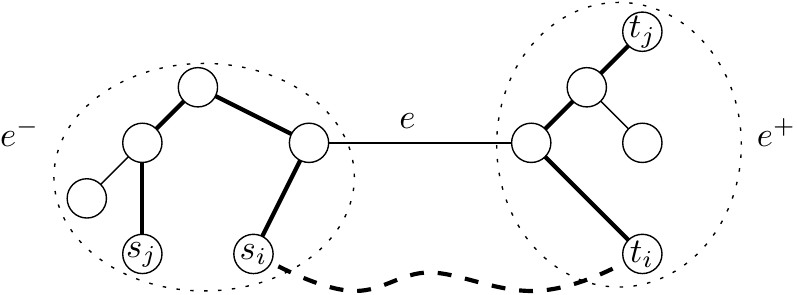}
\caption{$i\in U, j\in U$ and $T_j^i \neq O_j$. Then $e$ cannot be traversed in the path $T_j^i$.}
\label{case21}
\end{minipage}
\end{figure}
\begin{figure}
\centering
\begin{minipage}[h]{.48\textwidth}
\centering
\includegraphics[width=\linewidth]{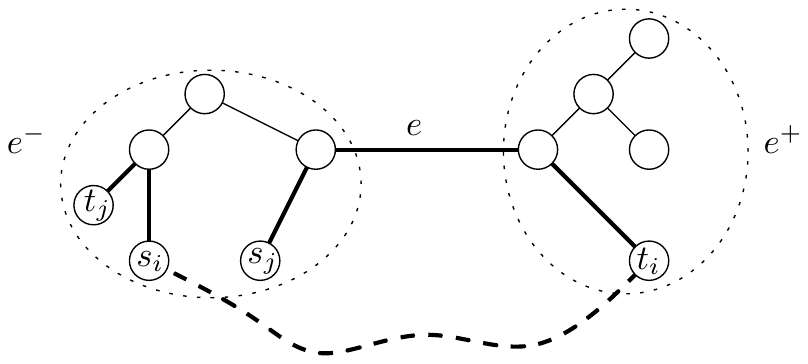}
\caption{$i\in U, j\notin U$, and $T_j^i \neq O_j$. Then $e$ can be traversed in the path $T_j^i$.}
\label{case23}
\end{minipage}%
\hfill
\begin{minipage}[h]{.48\textwidth}
\centering
\includegraphics[width=\linewidth]{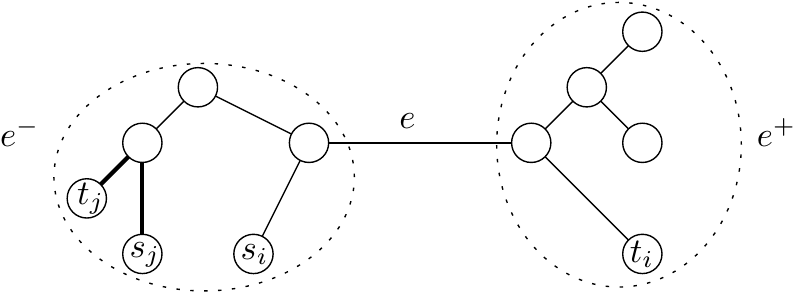}
\caption{$i\in U, j\notin U$, and $T_j^i = O_j$. Then $e$ cannot be traversed in the path $T_j^i$.}
\label{case24}
\end{minipage}
\end{figure}

Finally, to explain the second sum of (\ref{secondequation}), let $i\in U$.
The only possibilities are
that
%
\begin{itemize}
\item $j\in U$. Then $e$ cannot be traversed, since at least one of $(T_j^i)'$ or $(T_j^i)''$ is a \emph{$O$-cycle free} path that does not traverse $e$. See Fig. \ref{case21} for an illustration.
\item $j\notin U$ and $T_j^i \neq O_j$. Then $e$ can be traversed, since $s_j$ and $t_j$ are in the same connected component of $E(O)\setminus \{e\}$, but $s_i$ and $t_i$ are in different ones. See Fig. \ref{case23} for an illustration.
\item $j\notin U$ and $T_j^i = O_j$. Then $e$ cannot be traversed, since $s_j$ and $t_j$ are in the same connected component of $E(O)\setminus \{e\}$ and we just take the direct path between them, which does not traverse $e$. See Fig. \ref{case24} for an illustration.
\end{itemize}
Let $o_i(U)$ be the number of $j\notin U$ with $T_j^i \neq O_j$. 
Then, as we can see, $e$ is traversed at most $o_i(U)\leq n-\vert U \vert$ times.
This explain the second sum of (\ref{secondequation}) and finishes the proof of Lemma \ref{secondlemma}.

\end{proof}

Theorem \ref{mainthm} follows directly if $E(O)$ is connected but $O^n$ is empty
by Lemma \ref{secondlemma} and Lemma \ref{boundlemma}.
The following lemma handles the last case we have left to analyze, which is when
$E(O)$ is not a connected tree. This, together with Lemma~\ref{boundlemma},
finishes the proof of Theorem~\ref{mainthm}.

\begin{lemma}
\label{thirdlemma}
Let $E(O)=C_1 \sqcup \cdots \sqcup C_q$, with each $C_m$ being a connected component of $E(O)$.
Let $R_m$ be the set of players $j$ with $s_j,t_j \in C_m$. Then for a player $i\in R_k$ 
\begin{equation}
\label{thirdequation}
\Phi(N) \leq \Phi(T^i) \leq \sum_{U \subset \{1,\ldots,n\} \atop i\in U} H_n \vert N_U \vert + \sum_{U \subset R_k \atop i\in U} H_{o_i(U)} \vert O_U \vert + \sum_{\substack{U \subset R_m \text{ for some } m \\
 i\notin U}} H_{\vert U \vert} \vert O_U \vert \enspace ,
\end{equation}
with $o_i(U)\leq \vert T_k \vert -\vert U \vert \leq n - \vert U \vert$.
\end{lemma}
\begin{proof}
Since the initial part of the proof is exactly the same as the proof of Lemma \ref{firstlemma} and Lemma \ref{secondlemma}, we only prove that the cost $c_e$ of every edge $e$ in $T^i$ is accounted for with at least coefficient $H_{k_e(T^i)}$ in the right hand side of (\ref{thirdequation}).
In particular, we just look at edges that are only present in steps $1'$ and $3'$ of the definition of $T_j^i$, since an edge $e\in O_U$ that also belongs to $N_i$ has its cost already accounted for in the first sum.

To explain the second and third sum, let $U\subset \{1,\ldots,n\}$ and $e\in O_U$.
Notice that if $U\not\subset R_m$ for every $m$, then $O_U$ is the empty set and $e$ does not contribute anything to $\Phi(T^i)$.
We begin by looking at the second sum.

Notice that since $i\in R_k$, the only possibility to have $i\in U$ is that $U \subset R_k$.
By the definition of $T^i$ the players $j\in R_m$, $m\neq k$ use the path $O_j$, which does not traverse $e$.
With the exact same reasoning of Lemma \ref{secondlemma}, by looking at all the possibilities of where $s_i,t_i,s_j$ and $t_j$ can be in $C_k$, we can see that $e$ can be traversed by player $j\in R_k$ only if $j\notin U$ and $T^i \neq O_j$.
If we then define the number of players $j\in T_k$ with this property to be $o_i(U) \leq \vert T_k \vert -\vert U \vert \leq n - \vert U \vert$, the second sum in the right hand side of (\ref{thirdequation}) is explained.

Finally, for the third sum, we fix $i\notin U$ and look at the cases $U\subset R_k$ and $U\subset R_m$, $m\neq k$ separately.

Suppose first that $U\subset R_k$.
By the definition of $T^i$ the players $j\in R_m$, $m\neq k$ use the path $O_j$, which does not traverse $e$.
With the exact same reasoning of Lemma \ref{secondlemma}, by looking at all the possibilities of where $s_i,t_i,s_j$ and $t_j$ can be in $C_k$, we can see that $e$ can be traversed by player $j\in R_k$ only if $j\in U$.
That is, by at most $\vert U \vert$ players.
This explains the third sum for the case $U\subset R_k$.

We now look at the case $U\subset R_m$, $m\neq k$.
By the definition of $T^i$, players $j\in R_l$, $l\neq m$ do not traverse $e$, since they only use edges of $C_l$ (if $l\neq k$) or edges of $C_k$ and of $N_i$ (if $l=k$).
Players $j\in R_m$ use the path $O_j$, and by the definition of $O_U$ exactly $\vert U \vert$ players traverse $e$.
This explains the third sum for the case $U\subset R_m$, $m\neq k$, which finishes the proof.

\end{proof}
%
%
%


\section{Conclusion}

In this paper we improved the upper bound on price of stability of undirected 
network design games by analyzing potential minima and their properties. 
We hope that similar analysis can be applied to multicast games to obtain much
better asymptotic for the upper bound. 
It is known that bounding the cost of
potential minima cannot provide an upper bound on the price of stability better than
$\Theta(\sqrt{\log\log n})$ \cite{popos}. 
It remains an open question, whether
$\Theta(\sqrt{\log \log n})$ can actually be achieved.

\vspace{0.5cm}

\noindent\textbf{Acknowledgements.} We are grateful to Rati Gelashvili for
valuable discussions and remarks. This work has been partially supported by the
Swiss National Science Foundation (SNF) under the grant number
200021\_143323/1.

\bibliographystyle{plain}
\bibliography{conference_paper.bbl}

\end{document}